\documentclass[runningheads]{llncs}
\usepackage[T1]{fontenc}
\usepackage{graphicx}
\usepackage{booktabs}
\usepackage{array}
\usepackage{siunitx}
\usepackage{amsmath}
\usepackage{enumitem}
\usepackage{amssymb}
\usepackage{multirow}
\usepackage{xcolor}

\begin{document}

\title{CRSF: Collusion-Resilient Privacy-Preserving Sensor Fusion with Byzantine-Robust Participation}

\author{Chao Yin\inst{1,2} \and
Haihong Tian\inst{3} \and
Zheng Yang\inst{3}\and
Haibin Zhang\inst{4,5} \and
Fabio Massacci\inst{1,6} \and
Chenglu Jin\inst{2}}

\authorrunning{C. Yin et al.}
\institute{Vrije Universiteit Amsterdam, Amsterdam, the Netherlands \\ \email{c.yin@vu.nl, fabio.massacci@ieee.org} \and
Centrum Wiskunde \& Informatica, Amsterdam, the Netherlands \\
\email{chenglu.jin@cwi.nl} \and Southwest University, Chongqing, China \\ \email{tianhaihong@email.swu.edu.cn, youngzheng@swu.edu.cn} \and Yangtze Delta Region Institute of Tsinghua University, Zhejiang, China \\ \email{bchainzhang@aliyun.com} \and Jiaxing Key Laboratory of Artificial Intelligence and Cyber Resilience, Zhejiang, China \and
University of Trento, Trento, Italy}
\maketitle
\begin{abstract}

Privacy-preserving sensor fusion enables an untrusted server to compute an aggregate result over distributed sensor measurements without learning either individual inputs or the final output. Recent garbled-circuit-based protocols provide an efficient realization of this functionality in a sensor--server--client architecture, but remain vulnerable to sensor--server collusion and Byzantine manipulation of sensor participation. These weaknesses can compromise honest-sensor privacy, incorrectly exclude honest sensors, and corrupt the computed fusion result, thereby undermining the security guarantees expected from the protocol.

We present CRSF, a collusion-resilient sensor-fusion protocol that addresses these weaknesses while providing privacy, correctness with explicit abort, and liveness. CRSF introduces a Practical Byzantine Fault Tolerance (PBFT)-based agreement phase for sensor submissions and uses server-specific, status-dependent label release with threshold protection of circuit-input labels. This design prevents any Byzantine server from unilaterally manipulating sensor participation and prevents any admissible sensor-server coalition from obtaining enough secret material to compromise honest-sensor privacy.

We implement CRSF and compare its online execution time with the most relevant state-of-the-art baseline. Our Google Cloud evaluation measures the total computation and communication cost of the online protocol under fault-free and representative faulty executions. Across a range of fault-tolerant fusion circuits and up to 261 sensors, CRSF demonstrates a highly practical trade-off between robust security and protocol performance.

\end{abstract}
\newcommand{\numOfSensor}{n}
\newcommand{\sensorSet}{\mathcal{S}}
\newcommand{\sensori}[1]{s_{#1}}
\newcommand{\serverA}{\mathcal{A}}
\newcommand{\serverB}{\mathcal{B}}
\newcommand{\singleserver}{\mathsf{Srv}}
\newcommand{\client}{\mathcal{C}}

\newcommand{\adversary}{\mathit{Adv}}
\newcommand{\simulator}{\mathit{Sim}}

\newcommand{\sensorLabel}[3]{\sigma_{#1,#2}^{#3}}
\newcommand{\sensorLabelsingleserver}[3]{\sigma_{#1,#2}^{#3}}

\newcommand{\sLabel}{\sigma}
\newcommand{\cLabel}{C}
\newcommand{\Label}{W}

\newcommand{\clientLabel}[4]{C^{#3,(#4)}_{#1,#2}}
\newcommand{\wireShare}[4]{\omega^{#4}_{#1,#2,#3}}
\newcommand{\wireLabel}[3]{\omega_{#1,#2}^{#3}}
\newcommand{\circLabel}[3]{W_{#1,#2}^{#3}}
\newcommand{\zLabel}[3]{Z_{#1,#2}^{#3}}

\newcommand{\honest}{\mathsf{Hon}}
\newcommand{\malicious}{\mathsf{Mal}}

\newcommand{\lsb}{Lsb}
\newcommand{\FaultySensors}{\sigma_f}
\newcommand{\Sign}{Sig}
\newcommand{\corruptedParty}{\mathcal{M}}
\newcommand{\Enc}{\text{En}}
\newcommand{\Dec}{\text{De}}

\newcommand{\offset}{\Delta}

\newcommand{\share}{share}
\newcommand{\server}[1]{Ser_{#1}}

\newcommand{\simCircLabel}[3]{\widetilde{W}_{#1,#2}^{#3}}
\newcommand{\circShare}[4]{W_{#1,#2}^{#3,(#4)}}
\newcommand{\simCircShare}[4]{\widetilde{W}_{#1,#2}^{#3,(#4)}}

\section{Introduction}\label{sec:introduction}

Sensor fusion is a fundamental task in cyber-physical systems, distributed sensing, and safety-critical monitoring applications. By combining readings from multiple sensors, sensor fusion provides a view of the physical environment that is more accurate, robust, and actionable than any individual measurement alone. Many sensor-fusion deployments can be abstracted as a three-role architecture, as seen in mobile crowdsensing, participatory sensing, and cloud-assisted sensing systems~\cite{jin2024pg,wang2020privacy,liu2016data,de2011short,cornelius2008anonysense,de2009privacy}. Distributed sensors or mobile devices contribute measurements, a client or application specifies the fusion task and consumes the result, and an intermediate service platform collects, processes, or stores the data. This architecture imposes stringent security requirements: sensor readings may reveal sensitive information and should remain private from both the server and the client; corrupted sensors should not be able to arbitrarily bias or disrupt the computation; and, under the intended availability assumptions, the client should reliably obtain the fusion result.

Prior work on three-role sensing and aggregation systems has studied related security requirements from different angles, including privacy-preserving aggregation over sensor or mobile-device inputs, verifiable aggregation in the presence of corrupted aggregators or sensors, and robustness against false data or pollution attacks~\cite{przydatek2003sia,he2007pda,conti2009privacy,roy2014secure,nath2009secure,zhang2016privacy}. In the more specific setting of privacy-preserving sensor fusion, PG~\cite{jin2024pg} is a garbled-circuit-based protocol that combines input privacy with robustness against pollution attacks, and provides guaranteed output delivery under a bounded Byzantine model. To the best of our knowledge, PG is the first construction to simultaneously achieve input privacy, Byzantine fault tolerance, and guaranteed output delivery in a three-role sensor-fusion architecture. PG is therefore a particularly important baseline in this direction.

However, the security guarantees of PG do not extend to sensor-server collusion, where corrupted sensors and a Byzantine server combine their protocol views. This represents a significant gap, as collusion between data providers and aggregation-side entities is a well-documented threat in sensing systems~\cite{shi2010prisense,gunther2014privacy,yao2014privacy,palazzo2024privacy}. Such collusion is highly realistic in scenarios where a single adversary simultaneously compromises both a subset of sensors and a central server. This gap motivates us to revisit privacy-preserving sensor fusion under a stronger adversarial model, which accounts for both sensor-server collusion and Byzantine server behavior, while preserving the core properties of the existing baseline.

We identify two limitations of PG under this model. First, once sensor-server collusion is allowed, a corrupted sensor and a Byzantine server may combine their views to obtain complementary circuit-input labels on a circuit-input wire, which can compromise honest-sensor privacy and fusion result integrity. Second, PG relies on a single server to determine which sensor submissions will be used in data fusion. A Byzantine server can therefore manipulate the effective participating set of the sensors, for example by excluding honest sensors while retaining malicious ones, and thereby invalidate the fault-threshold premise required by the underlying fault-tolerant algorithm. These observations reveal two design requirements for strengthening garbled-circuit-based sensor fusion: participation decisions must be correctly made, and circuit-input labels must be protected against sensor-server coalitions.

To meet these requirements, we propose CRSF, a collusion-resilient sensor-fusion protocol that follows the garbled-circuit workflow, like PG~\cite{jin2024pg}, while strengthening it against sensor-server collusion and Byzantine manipulation of participation. CRSF extends PG by introducing Byzantine-robust agreement on sensor submissions with threshold protection of circuit-input labels. Our work makes the following contributions:

\begin{enumerate}

\item We identify two limitations in PG's trust assumptions: a single server unilaterally dictates sensor participation, and sensor-server collusion exposes secret circuit-input labels.

\item We present CRSF, a collusion-resilient garbled-circuit-based sensor-fusion protocol preserving the client-sensor-server workflow while resisting Byzantine participation manipulation and sensor-server collusion.

\item We introduce two protocol-level mechanisms for CRSF: agreement about sensor participation before status-dependent label release, and threshold protection of downstream circuit-input labels before fusion evaluation.

\item We prove that CRSF provides privacy against corrupted sensors colluding with one Byzantine server, and achieves correctness with explicit abort and liveness under our threat model.

\item We implement CRSF and evaluate its online performance on Google Cloud under both fault-free and representative faulty executions. The results show that the added sensor participation agreement and label reconstruction mechanisms incur moderate overhead, with total measured phase time remaining around one second at the largest tested instances.

\end{enumerate}

\paragraph{Paper organization.}
Section~\ref{sec:systemModel} presents the system and threat model. Section~\ref{sec:preliminaries} presents
preliminaries and reviews the PG baseline. Section~\ref{sec:pg-limitations} analyzes the limitations of PG under our threat model. Section~\ref{sec:Ourscheme} describes the CRSF
protocol. Section~\ref{sec:security-guarantees} proves its security. Section~\ref{sec:evaluation} presents the
experimental evaluation. Section~\ref{sec:RelatedWork} discusses related work, and Section~\ref{sec:conclusion}
concludes the paper.

\section{System and Threat Model}\label{sec:systemModel}

\noindent\textbf{System Model.}
We consider a sensor fusion system that involves three roles: a client, a set of $n$ sensors ($\{\sensori{i}\}_{i\in[0,n-1]}$) and four servers \(\server{1},\server{2},\server{3},\server{4}\). The client aims to obtain the output of a fusion function $f$ evaluated over the private measurements of the sensors, without receiving those measurements in the clear. During the setup phase, the client may interact with the sensors and generate the cryptographic material required for the protocol. In the online phase, each sensor $\sensori{i}$ locally encodes its measurement and sends the resulting encoded messages to the servers. The servers then perform the fusion computation on the encoded inputs and send the encoded outputs to the client, without learning the plaintext sensor measurements or the plaintext fusion result. During this phase, the client communicates only with the servers and is the only party that can decode the final output. Sensors are mutually independent and do not communicate with each other.

\vspace{2mm}
\noindent\textbf{Communication Model.} We assume sensor-to-server links are synchronous with a known delay bound, so that a submission sent by an honest sensor reaches every honest server before the submission deadline $T_{sub}$. Inter-server communication is only assumed to be partially synchronous, as required by PBFT.

\vspace{2mm}
\noindent\textbf{Threat Model.}
We consider a probabilistic polynomial-time adversary corrupting multiple sensors and at most one server. Corrupted sensors may arbitrarily deviate, submit malformed labels, and collude with the corrupted server. This server may exhibit Byzantine behaviors, such as sending inconsistent messages. The client honestly follows the protocol. The number of malicious sensors must remain below the underlying Marzullo-style algorithm's fault threshold (e.g., fewer than $1/3$ or $1/2$ of all sensors~\cite{jin2024pg}). Unavailable sensors also count toward this threshold.

\vspace{2mm}
\noindent\textbf{Relation to Existing Models.}
This three-role abstraction follows a common pattern in sensing and aggregation systems~\cite{wang2020privacy,de2013participatory,nath2009secure}. The collusion threat considered here is also consistent with prior work on sensing and aggregation security: existing studies have shown that privacy and integrity guarantees may fail when a central service, aggregator, or server combines its view with corrupted data providers, users, or sensors~\cite{przydatek2003sia,shi2010prisense,gunther2014privacy,palazzo2024privacy}. Our model instantiates this line of threat models in garbled-circuit-based sensor fusion, where sensor-server collusion creates additional risks for honest-sensor privacy.

\vspace{2mm}
\noindent\textbf{Security and Robustness Goals.} CRSF aims to achieve the following properties under the threat model considered in this work.
\begin{enumerate}
\item \textit{Privacy.}
Honest sensors' private measurements remain hidden from every individual server, even under sensor-server collusion within the stated corruption bounds. The client learns only the protocol output and the protocol-visible status information needed for output delivery.

\item \textit{Correctness.}
An honest client either outputs $\bot$ or accepts only the correct output of the fusion function evaluated on the protocol-determined sensor inputs and default replacements.

\item \textit{Liveness.}
An honest client eventually receives either the correct output or \(\bot\) under the stated assumptions.

\item \textit{Byzantine-robust agreement on participation.}
All honest servers agree on a common submission-level outcome for each sensor, namely whether to accept a submitted vector or exclude the sensor.

\end{enumerate}

\section{Preliminaries}\label{sec:preliminaries}

\noindent\textbf{Garbled circuits.}
Garbled circuits, introduced by Yao~\cite{yao1986generate} and later formalized in simulation-based frameworks~\cite{lindell2009proof,bellare2012foundations}, are a standard primitive for secure computation. We use the standard syntax of a garbling scheme for Boolean circuits. A garbling scheme is a tuple of probabilistic polynomial-time algorithms
$
\mathsf{GC}=(\mathsf{Gb},\mathsf{En},\mathsf{Ev},\mathsf{De}).
$
On input a security parameter $1^\kappa$ and a Boolean circuit $C$, the garbling algorithm samples
$
(\widehat C,\mathsf{e},\mathsf{d}) \leftarrow \mathsf{Gb}(1^\kappa,C),
$
where $\widehat C$ is the garbled circuit, $\mathsf{e}$ is encoding information, and $\mathsf{d}$ is decoding information. For a circuit with input wires $w_1,\ldots,w_m$, a projective garbling scheme associates each input wire $w_i$ with two labels $L_{w_i}^0$ and $L_{w_i}^1$, representing the semantic values 0 and 1, respectively. Given an input vector $x=(x_1,\ldots,x_m)\in \{0,1\}^m$, the encoding algorithm outputs the encoded input vector
$
X \leftarrow \mathsf{En}(\mathsf{e},x),
$
which contains exactly one label for each input wire. The evaluation algorithm computes
$
Y \leftarrow \mathsf{Ev}(\widehat C,X),
$
where $Y$ is a vector of garbled output labels. The decoding algorithm then maps these labels to a plaintext output,
$
y \leftarrow \mathsf{De}(\mathsf{d},Y).
$
The decoded value satisfies $y=C(x)$ except with negligible probability.

Garbling schemes are commonly studied with several security notions, including privacy, obliviousness, and authenticity~\cite{bellare2012foundations}. Let $\Phi(C)$ denote the public leakage associated with garbling $C$, such as the circuit size, input and output lengths, or topology, depending on the concrete scheme and application.

\begin{itemize}
    \item \noindent\textit{Obliviousness.} A garbling scheme $\mathsf{GC}=(\mathsf{Gb},\mathsf{En},\mathsf{Ev},\mathsf{De})$ satisfies obliviousness if there exists a PPT simulator $\mathsf{Sim}^{\mathsf{obv}}_{\mathsf{GC}}$ such that, for every circuit $C$ and input $x$, $(\widehat C,X)\approx_c \mathsf{Sim}^{\mathsf{obv}}_{\mathsf{GC}}(1^\kappa,\Phi(C))$, where $(\widehat C,\mathsf{e},\mathsf{d})\leftarrow\mathsf{Gb}(1^\kappa,C)$ and $X\leftarrow\mathsf{En}(\mathsf{e},x)$.
    \item \noindent\textit{Authenticity.} A garbling scheme satisfies authenticity if, for every PPT adversary $\mathcal A$, every circuit $C$, and every input $x$, after $(\widehat C,\mathsf{e},\mathsf{d})\leftarrow\mathsf{Gb}(1^\kappa,C)$, $X\leftarrow\mathsf{En}(\mathsf{e},x)$, and $Y\leftarrow\mathsf{Ev}(\widehat C,X)$, the probability that $\mathcal A(1^\kappa,\widehat C,X)$ outputs $Y'\neq Y$ such that $\mathsf{De}(\mathsf{d},Y')\neq\bot$ is negligible in $\kappa$.
\end{itemize}

\noindent\textbf{Shamir secret sharing.}
A $(t,n)$-threshold Shamir secret sharing scheme~\cite{shamir1979share} distributes a secret $s \in \mathbb{F}$ among $n$ parties over a finite field $\mathbb{F}$, where $|\mathbb{F}| > n$. The dealer samples a random polynomial
$
p(X)=s+a_1X+\cdots+a_{t-1}X^{t-1}
$
of degree at most $t-1$, where $a_1,\ldots,a_{t-1}$ are chosen independently and uniformly from $\mathbb{F}$. Each party $i$ receives the share $p(\alpha_i)$, where $\alpha_1,\ldots,\alpha_n$ are distinct nonzero field elements. Any set of $t$ shares can reconstruct $s=p(0)$ by interpolation, whereas any set of fewer than $t$ shares is statistically independent of $s$.

\vspace{2mm}
\noindent\textbf{Practical Byzantine Fault Tolerance.} Practical Byzantine Fault Tolerance (PBFT)~\cite{castro2002practical} is a Byzantine fault-tolerant state-machine replication protocol. For $n$ replicas, PBFT tolerates up to $f$ Byzantine replicas provided that $n \geq 3f+1$. In each consensus instance, a primary proposes a value in the pre-prepare phase, and replicas then use prepare and commit phases to form Byzantine quorums. PBFT provides safety and liveness under its standard assumptions. Safety requires that no two honest replicas decide different values in the same consensus instance. Liveness holds under partial synchrony, authenticated communication, and at most $f$ Byzantine replicas. If the current primary is faulty or fails to make progress, honest replicas can initiate a view change, moving the protocol to a new view with a different primary while preserving previously prepared information. Thus, once the network satisfies the synchrony condition, a faulty primary cannot indefinitely prevent progress. CRSF uses a PBFT-based submission agreement protocol rather than full state-machine replication. Specifically, the servers use PBFT's quorum and phase structure to agree on an evidence-supported submission-level outcome for each sensor.

\vspace{2mm}
\noindent\textbf{The PG Baseline.} PG~\cite{jin2024pg} is a privacy-preserving sensor-fusion protocol for the sensor-server-client architecture. It combines Marzullo-style fault-tolerant algorithms~\cite{marzullo1990tolerating,jin2024pg} with garbled circuits. In addition to privacy-preserving computation, PG aims to tolerate malicious sensor behavior and to ensure output delivery under its original threat model.

At a high level, PG separates circuit garbling from sensor-input garbling: the client prepares the garbled circuits used by the protocol, while each sensor garbles its own private measurement into sensor-side input labels. The server thus receives the garbled circuits from the client and the garbled sensor inputs from the sensors, and then evaluates the relevant gates as in standard garbled-circuit evaluation. To support malformed-input detection and liveness, PG introduces two special gate types, called \emph{checking gates} and \emph{filter gates}. These gates form two input-processing layers that are evaluated by the server before it evaluates the garbled fusion circuit. For clarity, we use circuit-input labels to denote the labels used to evaluate the garbled fusion function circuit, distinguishing them from the labels used as inputs to the checking gates and filter gates. Note that garbled sensor inputs, namely sensor submissions, are not the circuit-input labels used for fusion circuit evaluation.

Checking gates are used to check whether the garbled sensor inputs submitted by sensors are well-formed. Each checking gate is associated with one or two specific input wires and is constructed so that any combination of valid labels for these wires is evaluated to $0^\kappa$.
This allows the server to detect malformed sensor submissions from the outputs obtained by evaluating the checking gates, namely, a submission is malformed if the corresponding output is not a zero string. The server then informs the client which sensors are unavailable or fail to pass checking gate validation before proceeding to the filter gate layer.

Filter gates are then used to derive the circuit-input labels used in the subsequent fusion circuit evaluation. Each filter gate takes as input a sensor-side input label and a client-side label. The client releases the client-side label according to the server's request, based on whether the corresponding sensor submitted well-formed labels. The filter gate is designed so that, for a sensor that submits well-formed labels, the server obtains circuit-input labels with the same truth values as the submitted sensor-side input labels. Otherwise, the server obtains a designated default label corresponding to the semantic value~$1$, which represents the maximum-value replacement for sensors that fail to submit well-formed labels. More concretely, consider a filter gate with two input wires, one for the client-side input and one for the sensor-side input.

Since PG inherits the FreeXOR~\cite{kolesnikov2008improved} optimization, these circuit-input labels adopt a fixed global offset $\offset$, satisfying $\Label^1 = \Label^0 \oplus \offset$. The offset is used only for wires in the fusion circuit, including these circuit-input labels, whereas checking gates and the input wires of filter gates are garbled independently of $\offset$.

Section~\ref{sec:pg-limitations} analyzes the limitations of PG under the threat model of Section~\ref{sec:systemModel}.

\section{Limitations of PG}
\label{sec:pg-limitations}

PG is secure under its original threat model, where the server does not collude
with sensors and follows the prescribed participation procedure. Under the
stronger model considered in this work, two limitations arise. The first requires active sensor--server collusion; the second is an inherent weakness of the PG design that a Byzantine server can exploit without any sensor cooperation.

First, PG does not preserve input privacy once sensor--server collusion is
allowed. A Byzantine server colluding with a corrupted sensor can declare the
sensor to be honest during filter-gate evaluation, causing the client to release
the client-side label information for the honest branch. Since the corrupted
sensor can generate both valid sensor-side input labels for its own input
positions, the coalition can evaluate the corresponding filter gates under both
labels and obtain both valid circuit-input labels on the same fusion-circuit input
wire. XORing these labels reveals the FreeXOR offset $\Delta$. Once $\Delta$ is
known, any valid circuit-input label on a fusion-circuit wire reveals its
complementary label. The adversary is therefore no longer restricted to the
unique label assignment induced by the honest execution, which violates the intended input-privacy guarantee of PG. Also, knowing $\Delta$ would trivially allow a malicious server to flip the semantic meaning of the fusion circuit outputs.
Such output flipping additionally breaks the authenticity of the underlying garbling scheme: an adversary who knows $\Delta$ can produce $Y' \neq Y$ on the output wires such that $\mathsf{De}(\mathsf{d}, Y')$ decodes to a valid but semantically flipped result, violating the standard authenticity guarantee that no PPT adversary should produce a non-$\bot$ decoded output different from $C(x)$~\cite{bellare2012foundations}.

Second, and independently of sensor--server collusion, PG relies on a single server to determine which sensor submissions are
admitted to the fusion computation. A Byzantine server acting alone can therefore exclude
honest sensors while retaining corrupted ones, changing the effective
participating set and potentially violating the fault-threshold premise required
by the underlying fault-tolerant fusion algorithm.
This gap motivates the Byzantine-robust participation agreement goal introduced in this work (Section~\ref{sec:systemModel}).

The two limitations have distinct structural roots. In the first case, the root weakness is that a colluding server can obtain both circuit-input labels of a wire by evaluating the filter gate under both sensor-input labels from a compromised sensor. In the second, it is that the single server in PG can unilaterally decide which sensor submissions are admitted.

\vspace{-2mm}
\noindent\paragraph{Illustrative example.}
Consider a circuit with two input AND gates whose outputs feed a third AND gate. Sensor~$s_1$ provides the inputs to the first AND gate, and sensor~$s_2$ provides the inputs to the second. If $s_1$ colludes with the server, the coalition can obtain both valid labels for $s_1$'s input wires and recover the global FreeXOR offset $\Delta$. It can then derive complementary valid labels on other fusion-circuit wires, including those of $s_2$, and evaluate the second AND gate under alternative valid input-label combinations. Observing the resulting output labels reveals the semantic correspondence between $s_2$'s labels and their truth values. The output of the third AND gate can be flipped by XORing with $\Delta$. This illustrates why the original PG design does not provide privacy and correctness under sensor--server collusion.

These limitations motivate CRSF, which extends the PG-style garbled-circuit workflow while providing resilience against sensor--server collusion and Byzantine manipulation of sensor participation.

\section{The CRSF Protocol}\label{sec:Ourscheme}

A seemingly direct way to mitigate the label-exposure problem in PG is to use
OT-based label selection. In such a design, the client acts as the OT sender
and provides each sensor with only the sensor-side input labels corresponding to
that sensor's actual input bits. Thus, each sensor no longer obtains both valid
sensor-side input labels for any input position, which prevents a
sensor--server coalition from deriving both corresponding circuit-input labels
on the same fusion-circuit input wire through filter-gate evaluation.

However, this approach has two limitations in our setting. Architecturally, it requires each sensor to run online OT with the client in addition to submitting encoded inputs to the servers. This increases the online cost for both sensors and the client, and forces the client to participate in per-sensor label-transfer interactions, whereas the client is intended to remain lightweight during the online phase. Functionally, OT-based label selection addresses only the label-exposure problem; it does not prevent a Byzantine server from manipulating which sensor submissions are admitted to the fusion computation.

CRSF therefore follows a different design. Its design is guided by two
principles. First, an honest sensor submission
delivered to the honest servers should not be unilaterally excluded by a
Byzantine server, since such exclusion can invalidate the fault-tolerance
guarantee of the protocol. Second, for each input wire of the fusion circuit,
the servers should reconstruct only the circuit-input label determined by the
agreed sensor submission and the accepted status vector. Accordingly, CRSF uses
PBFT-based agreement among four servers to establish a common batch record for
the current set of sensor submissions. CRSF also incorporates secret sharing
into the filter gate design: each filter gate outputs a server-specific share
rather than the complete circuit-input label.

We first present the setup assumptions and offline preparation in Section~\ref{sec:setupandpreparation}. The online protocol is then organized into two stages: Section~\ref{sec:submission-agreement} describes sensor submission collection and PBFT-based agreement, and Section~\ref{sec:labelreconstruction} presents the subsequent online evaluation phase.

\subsection{Setup and Offline Preparation}\label{sec:setupandpreparation}

\noindent\textbf{Setup and authentication material.}
The protocol uses a collision-resistant hash function $H:\{0,1\}^\ast\rightarrow\{0,1\}^{\kappa}$~\cite{rogaway2004cryptographic}, an EUF-CMA secure digital signature scheme $\mathsf{Sig}=(\mathsf{KeyGen},\mathsf{Sign},\mathsf{Verify})$~\cite{goldwasser1988digital}, and a keyed pseudorandom label-derivation function $\mathsf{PRG}$. Each sensor $\sensori{i}$ holds a signing key $\mathsf{sk}_{s_i}$, and its verification key $\mathsf{vk}_{s_i}$ is known to all servers. Each server $\server{h}$ holds a signing key $\mathsf{sk}_{\server{h}}$, and its verification key $\mathsf{vk}_{\server{h}}$ is known to the client and the other servers. Each sensor also shares a pairwise secret key with the client for deriving sensor-side input labels.

\vspace{2mm}
\noindent\textbf{Offline preparation.} The client prepares the garbled tables and label material required for the online protocol (detailed in Fig.~\ref{fig:crsf-offline-functions}). It distributes the garbled fusion circuit and checking gates to all servers, but each server receives only its server-specific filter gates.

\begin{figure*}[t]
\centering
\scriptsize
\setlength{\fboxsep}{6pt}
\fbox{%
\begin{minipage}{0.96\textwidth}
\begin{minipage}[t]{0.52\textwidth}
\raggedright
\setlength{\parindent}{0pt}

\textbf{procedure} \textsc{PrepCRSF}$(C_{\mathsf{fus}},\{\ell_i\}_{i \in [n]})$:\par
\vspace{2pt}\par
\hspace*{1em}
$\offset \leftarrow \{0,1\}^{\kappa-1}1$\par
\hspace*{1em} \textbf{for} each sensor $\sensori{i}$ \textbf{do}\par
\hspace*{2em} establish seed $K_i$ with $\sensori{i}$\par
\hspace*{2em} \textbf{for} $j\in\{0,\ldots,\ell_i-1\}$ and $b\in\{0,1\}$ \textbf{do}\par
\hspace*{3em} $\sensorLabel{i}{j}{b}\gets \mathsf{PRG}(K_i,i,j,b)$\par

\vspace{1pt}
\hspace*{1em} \textbf{for} each $i,j,h$ \textbf{do}\par
\hspace*{2em} sample $\clientLabel{i}{j}{\honest}{h}$ and $\clientLabel{i}{j}{\malicious}{h}$\par

\vspace{1pt}
\hspace*{1em} \textbf{for} each $i,j$ \textbf{do}\par
\hspace*{2em} $(W_{i,j}^{0},W_{i,j}^{1},\{W_{i,j}^{b,(h)}\}_{b,h})$ \par
\hspace*{2em} $\gets \textsc{ShareWire}(\offset)$\par

\vspace{1pt}
\hspace*{1em} \textbf{for} each $i,j,h$ \textbf{do}\par
\hspace*{2em} $\widehat F_{i,j}^{h}\gets \textsc{GbFilt}(\clientLabel{i}{j}{\honest}{h},\clientLabel{i}{j}{\malicious}{h},$\par
\hspace*{3em} $\sensorLabel{i}{j}{0},\sensorLabel{i}{j}{1}, W_{i,j}^{0,(h)},W_{i,j}^{1,(h)})$\par

\hspace*{1em} \textbf{for} each $\sensori{i}$ and $t\in\{0,\ldots,\lceil \ell_i/2\rceil-1\}$ \textbf{do}\par
\hspace*{2em} $j\gets 2t$\par
\hspace*{2em} \textbf{if} $j+1<\ell_i$ \textbf{then}\par
\hspace*{3em} $\widehat G_{i,t}^{(1)}\gets \textsc{GbChk}(\sensorLabel{i}{j}{0},\sensorLabel{i}{j}{1},\sensorLabel{i}{j+1}{0},\sensorLabel{i}{j+1}{1})$\par
\hspace*{3em} $\widehat G_{i,t}^{(2)} \gets $\par \hspace*{5em} $\textsc{GbChk}(W_{i,j}^{0},W_{i,j}^{1},W_{i,j+1}^{0},W_{i,j+1}^{1})$\par
\hspace*{2em} \textbf{else}\par
\hspace*{3em} $\widehat G_{i,t}^{(1)}\gets \textsc{GbChk1}(\sensorLabel{i}{j}{0},\sensorLabel{i}{j}{1})$\par
\hspace*{3em} $\widehat G_{i,t}^{(2)}\gets \textsc{GbChk1}(W_{i,j}^{0},W_{i,j}^{1})$\par

\vspace{1pt}

\hspace*{1em} $\widehat C_{\mathsf{fus}}\leftarrow \mathsf{Gb}(C_{\mathsf{fus}};\offset,\{W_{i,j}^0,W_{i,j}^1\}_{i,j})$\par
\hspace*{1em} send $\widehat C_{\mathsf{fus}}$ and $\{\widehat G_{i,t}^{(1)},\widehat G_{i,t}^{(2)}\}_{i,t}$ to all servers \par
\hspace*{1em} \textbf{for} each server $\server{h}$ \textbf{do}\par
\hspace*{2em} send $\{\widehat F_{i,j}^{h}\}_{i,j}$ to $\server{h}$\par

\end{minipage}
\hfill
\begin{minipage}[t]{0.5\textwidth}
\raggedright
\setlength{\parindent}{0pt}

\textbf{procedure} \textsc{ShareWire}$(\offset)$:\par
\vspace{2pt}
\hspace*{1em} sample $W^0\gets\{0,1\}^{\kappa}$\par
\hspace*{1em} set $W^1\gets W^0\oplus\offset$\par
\hspace*{1em} \textbf{for} $b\in\{0,1\}$ \textbf{do}\par
\hspace*{2em} $(W^{b,(1)},W^{b,(2)},W^{b,(3)},W^{b,(4)})$\par
\hspace*{3em} $\gets \mathsf{Share}_{3,4}^{\mathrm{byte}}(W^b)$\par
\hspace*{1em} \textbf{return} $(W^0,W^1,\{W^{b,(h)}\}_{b,h})$\par

\vspace{8pt}
\textbf{procedure} \textsc{GbFilt}$(C^{\mathsf{Hon}},C^{\mathsf{Mal}},S^0,S^1,\omega^0,\omega^1)$:\par
\vspace{2pt}
\hspace*{1em} define mapping $\mathcal{T}_F$ by\par
\hspace*{2em} $\mathcal{T}_F(C^{\mathsf{Hon}},S^0)=\omega^0$\par
\hspace*{2em} $\mathcal{T}_F(C^{\mathsf{Hon}},S^1)=\omega^1$\par
\hspace*{2em} $\mathcal{T}_F(C^{\mathsf{Mal}},\varnothing)=\omega^1$\par
\hspace*{1em} \textbf{return} $\mathsf{GbGate}(\mathcal{T}_F)$\par

\vspace{8pt}
\textbf{procedure} \textsc{GbChk}$(A^0,A^1,B^0,B^1)$:\par
\vspace{2pt}
\hspace*{1em} define mapping $\mathcal{T}_G$ by\par
\hspace*{2em} \textbf{for} $b,b'\in\{0,1\}$ \textbf{do}\par
\hspace*{3em} $\mathcal{T}_G(A^b,B^{b'})=0^\kappa$\par
\hspace*{1em} \textbf{return} $\mathsf{GbGate}(\mathcal{T}_G)$\par

\vspace{8pt}
\textbf{procedure} \textsc{GbChk1}$(A^0,A^1)$:\par
\vspace{2pt}
\hspace*{1em} define mapping $\mathcal{T}_{G1}$ by\par
\hspace*{2em} $\mathcal{T}_{G1}(A^0)=0^\kappa$\par
\hspace*{2em} $\mathcal{T}_{G1}(A^1)=0^\kappa$\par
\hspace*{1em} \textbf{return} $\mathsf{GbGate}(\mathcal{T}_{G1})$\par

\end{minipage}
\end{minipage}%
}
\caption{CRSF offline preparation. The procedure $\mathsf{GbGate}$ garbles a special gate according to the specified input-label-to-output-label mapping.}
\label{fig:crsf-offline-functions}
\end{figure*}

\subsection{Agreement on Sensor Submissions} \label{sec:submission-agreement}

Online phase of CRSF starts with a PBFT-based evidence-supported agreement phase among the four servers. This phase establishes a common submission-level record for the current session. For clarity, we describe the agreement rule for an individual sensor submission. For each sensor $\sensori{i}$, the servers determine a common submission-level outcome $o_i \in \{(\mathsf{Hon},r_i),\mathsf{Mal}\}$ before label validation and fusion evaluation. The labels $\mathsf{Hon}$ and $\mathsf{Mal}$ denote protocol-level acceptance states: $\mathsf{Hon}$ means that the servers agree to use $r_i$ as the submitted vector for $\sensori{i}$, whereas $\mathsf{Mal}$ means that $\sensori{i}$ is excluded from the current session. In implementation, the servers batch the agreement for all sensors in the same session by agreeing on an outcome vector $O=(o_1,\ldots,o_n)$. Let $f=1$. The four-server setting tolerates one Byzantine server, and each PBFT quorum has size $2f+1=3$.

\vspace{2mm}
\noindent\textbf{Sensor-side submission.}
Let $\mathsf{sid}$ be a session identifier.
For each target server $\server{h}$, let $r_i^{(h)}$ be the complete input-label vector submitted by $\sensori{i}$ to $\server{h}$, and let $d_i^{(h)}=H(r_i^{(h)})$. The sensor signs the session identifier, the sensor identity, the target server identity, and the digest as $\pi_i^{(h)}=\mathsf{Sign}_{{\mathsf{sk}_{\sensori{i}}}}(\mathsf{sid},i,h,d_i^{(h)})$. The submission sent to $\server{h}$ is $\mathsf{sub}_i^{(h)}=(r_i^{(h)},\pi_i^{(h)})$.

\vspace{2mm}
\noindent\textbf{Server-side agreement.}
\begin{enumerate}[leftmargin=*,itemsep=1pt,topsep=1pt]
\item \textbf{Local reports.}
After the sensor-submission deadline $T_{\mathsf{sub}}$, each server forms one signed local report for $\sensori{i}$. If $\server{h}$ received a valid submission before $T_{\mathsf{sub}}$, it creates a submission report $R_i^{(h)}=(\mathsf{Sub},h,\mathsf{sub}_i^{(h)},\rho_i^{(h)})$, where $\rho_i^{(h)}=\mathsf{Sign}_{{\mathsf{sk}_{\server{h}}}}$ $(\mathsf{sid},i,h,\mathsf{Sub},H(\mathsf{sub}_i^{(h)}))$. If $\server{h}$ did not receive a valid submission before $T_{\mathsf{sub}}$, it creates a no-submission report $R_i^{(h)}=(\mathsf{NoSub},h,T_{\mathsf{sub}},\rho_i^{(h)})$, where $\rho_i^{(h)}=\mathsf{Sign}_{{\mathsf{sk}_{\server{h}}}}(\mathsf{sid},i,h,\mathsf{NoSub},T_{\mathsf{sub}})$. The primary also forms its own local report, and each non-primary server sends its local report to the primary of the current view. The view identifies the current primary for this agreement instance.
\item \textbf{Evidence rule.}
The primary constructs an evidence set $E_i$ consisting of $2f+1=3$ valid reports from distinct servers. A report is valid if the server signature verifies, the server identity is distinct within $E_i$, and every embedded sensor submission verifies under $\mathsf{vk}_{\sensori{i}}$. The outcome is derived deterministically from $E_i$. If $E_i$ contains at least $f+1=2$ valid submission reports with the same input-label vector $r_i$, then the only valid outcome is $o_i=(\mathsf{Hon},r_i)$. Otherwise, the derived outcome is $o_i=\mathsf{Mal}$.
\item \textbf{PBFT-based agreement.}
\begin{itemize}[leftmargin=*,itemsep=2pt,topsep=2pt]
    \item \emph{Pre-prepare.} The primary broadcasts a pre-prepare proposal \\
$\langle \mathsf{PRE}$ $\mathsf{\text{-}PREPARE},\mathsf{sid},i,\mathsf{view},o_i,E_i\rangle_{\server{p}}$.
A server accepts a pre-prepare message only if the primary signature verifies, the instance identifier is correct, $E_i$ is a valid evidence set, and $o_i$ is exactly the outcome derived from $E_i$. Thus, a proposal for $\mathsf{Mal}$ is invalid if $E_i$ contains two matching valid submissions, and a proposal for $(\mathsf{Hon},r_i)$ is invalid if $E_i$ does not contain two matching valid submissions for $r_i$.
\item \emph{Prepare and commit.}
After accepting a valid pre-prepare proposal, a server broadcasts a prepare
message for the digest $H(o_i,E_i)$. A server becomes prepared for this digest
after it gets three matching prepare messages from distinct servers, including the pre-prepare message from the primary. Once
prepared, the server broadcasts a commit message for the same digest. It decides
$o_i$ only after obtaining three matching commit messages for this digest.

\end{itemize}
\end{enumerate}

If the primary fails to send a valid pre-prepare message before timeout, or if equivocation prevents a prepare quorum, backups invoke a view-change procedure. The new primary preserves any prepared value required by PBFT safety. Otherwise, it constructs a fresh evidence set from three distinct servers' reports and proposes the deterministically derived outcome.

Under PBFT safety, all honest servers decide the same evidence-supported outcome for each sensor: either $(\mathsf{Hon},r_i)$ for a common accepted input-label vector, or $\mathsf{Mal}$ for exclusion. In addition, the evidence rule prevents a delivered honest submission from being excluded. Under the sensor-delivery assumption, if $\sensori{i}$ is honest and participates, then all honest servers receive the same valid submission from $\sensori{i}$ before $T_{\mathsf{sub}}$. Since at most one server is Byzantine, every evidence set of three distinct valid reports contains at least two honest reports, and these reports carry the same submission vector and digest. Therefore, the deterministic evidence rule derives $o_i=(\mathsf{Hon},r_i)$, rather than $\mathsf{Mal}$.

\subsection{Status-Dependent Label Release and Reconstruction}\label{sec:labelreconstruction}

Following the PBFT-based agreement in Section~\ref{sec:submission-agreement}, the servers hold a common outcome vector $O=(o_1,\ldots,o_n)$. This vector identifies accepted submissions, but does not validate whether the included sensor-side input labels are well-formed. The first-layer checking gates perform this validation, yielding a status vector $u=(u_1,\ldots,u_n)$. Each entry determines whether the client releases the $\honest$ or $\malicious$ branch of the corresponding client-side filter-gate labels. The remainder of this section details the online evaluation phase: status derivation, server-specific label release, reconstruction, and final fusion-circuit evaluation.

\vspace{2mm}
\noindent\textbf{Online Evaluation.}
\begin{enumerate}[leftmargin=*,itemsep=2pt,topsep=2pt]
    \item \textbf{First-layer checking.}
For each sensor $\sensori{i}$, if $o_i=(\mathsf{Hon},r_i)$, each server evaluates the first-layer checking gates on the sensor-side input labels in $r_i$. If all checking-gate outputs are $0^\kappa$, the status of $\sensori{i}$ is set to $\honest$; otherwise, it is set to $\malicious$.  Thus, each server derives a status vector $u=(u_1,\ldots,u_n)$ according to
$$
u_i =
\begin{cases}
\honest, & \text{if } o_i=(\mathsf{Hon}, r_i) \text{ and } r_i \text{ passes the first-layer checking gates},\\
\malicious, & \text{otherwise.}
\end{cases}
$$

\item \textbf{Filter gate transformation.}
\begin{itemize}[leftmargin=*,itemsep=2pt,topsep=2pt]
    \item \emph{Label release.}
Each server sends the resulting status vector $u = (u_1,\ldots,u_n)$ to the client. The client accepts a status vector only if at least three servers report the same vector; otherwise, the client aborts. Let $u^\star$ denote this accepted majority status vector. Once such a $3$-of-$4$ majority status vector is obtained, the client sends to each server its server-specific client-side filter-gate labels corresponding to this majority vector. In particular, for each sensor and input position, the client releases only one branch of client-side labels and never releases both the $\honest$ and $\malicious$ branch labels.
\item \emph{Evaluation.} The servers then enter the filter-gate layer. For each server $\server{h}$, each sensor $\sensori{i}$, and each input position $j$, server $\server{h}$ evaluates its server-specific filter gate using the released client-side label and the corresponding sensor-side input label if $u_i^\star=\honest$, or using the $\malicious$ branch otherwise. The output is not a complete circuit-input label. Instead, it is a server-specific share of the circuit-input label determined by the accepted status branch and the corresponding sensor-side input label.
\end{itemize}
\item \textbf{Reconstruction with checking.} The servers exchange these shares and reconstruct candidate circuit-input labels from triples of shares. Since each complete circuit-input label is shared using byte-wise $(3,4)$ Shamir sharing, any three valid shares determine one candidate label. A server accepts a reconstructed candidate label set only if it passes the check of another layer checking gates, which we call the second-layer checking gates. If no candidate passes the check, the server aborts. All the checked and validated circuit-input label sets are used as the input to the fusion circuit evaluation.
\item \textbf{Fusion circuit evaluation.} Finally, each server evaluates the garbled fusion circuit on the validated circuit-input labels and obtains an output-label vector. The servers send their output-label vectors to the client.
\end{enumerate}

The client accepts an output-label vector only if it receives the same vector from at least three servers; otherwise, it aborts. It then decodes the accepted vector using its output decoding information and outputs the fusion result. This completes the online execution of CRSF.

\section{Security Analysis}
\label{sec:security-guarantees}

We now analyze the security guarantees provided by CRSF under the adversarial model of Section~\ref{sec:systemModel}. We assume that the client honestly generates the setup material and that public verification keys are distributed correctly. Cryptographically, we assume that $H$ is collision resistant, $\mathsf{Sig}$ is EUF-CMA secure, $\mathsf{PRG}$ is pseudorandom, the garbling scheme satisfies privacy, obliviousness, and authenticity, and the byte-wise $(3,4)$ Shamir sharing scheme satisfies correctness and privacy. For liveness, we additionally assume the standard PBFT liveness conditions, including authenticated communication, partial synchrony, and eventual delivery of messages among honest parties.

We begin by defining the corrupted coalition's view and the leakage available to the simulator, which are then used to state and prove the privacy theorem. We then establish correctness with explicit abort and liveness.

Let $\widehat C_{\mathsf{fus}}$ denote the common garbled fusion circuit, and let $\widehat{\mathcal G}^{(1)}$ and $\widehat{\mathcal G}^{(2)}$ denote the common first- and second-layer checking gates. All these common gates are known to all the servers. For each server $\server{h}$, let $\widehat{\mathcal F}_h$ denote its server-specific filter gates. We write $\widehat{\mathcal F}_q$ for the filter gates in the corrupted server's view. Let $\Phi_{\mathsf{fus}}$ denote the public structural leakage of the garbled fusion circuit, such as its topology, size, input length, and output length. Let $\Phi_{\mathsf{aux}}$ denote the public structural leakage of the checking and filter layers, including the number and arity of checking gates and filter gates, the input lengths $\{\ell_i\}_{i \in [n]}$, and the public wiring pattern between these auxiliary gates and the fusion-circuit input wires.

Fix an adversary $\mathcal A$ that corrupts server $\server{q}$ and a set of sensors. Let $\mathsf{View}_{\mathcal A,q}$ denote the joint view of the corrupted coalition. This view consists of the corrupted parties' inputs, randomness, and local states; the corrupted sensors' submitted labels; the messages sent and received by $\server{q}$; the PBFT-based agreement transcript visible to $\server{q}$; the status-vector and client-side label-release transcript; the garbled artifacts visible to $\server{q}$, namely $\widehat C_{\mathsf{fus}}$, $\widehat{\mathcal G}^{(1)}$, $\widehat{\mathcal G}^{(2)}$, and $\widehat{\mathcal F}_q$; the checking- and filter-gate outputs obtained by $\server{q}$; the reconstruction shares received by $\server{q}$; the reconstructed circuit-input labels visible to $\server{q}$; and, if the protocol reaches the final evaluation phase, the output-label vector obtained from evaluating the garbled fusion circuit. The client-side decoding information is not included in $\mathsf{View}_{\mathcal A,q}$.

\begin{definition}[Coalition leakage]
For an adversary $\mathcal A$ corrupting server $\server{q}$ and a set of sensors, define
\[
L^{\mathsf{coal}}_{\mathcal A,q}
=
(q,\mathsf{inp}_{\mathcal A,q},\mathsf{rand}_{\mathcal A,q},
\Phi_{\mathsf{fus}},\Phi_{\mathsf{aux}},\mathsf{type}(O),u^\star,\mathsf{abortinfo}_q).
\]
Here $\mathsf{inp}_{\mathcal A,q}$ and $\mathsf{rand}_{\mathcal A,q}$ are the corrupted parties' inputs and randomness, $\mathsf{type}(O)$ keeps only the $\mathsf{Hon}/\mathsf{Mal}$ type of each agreed outcome and removes every accepted vector $r_i$, and $\mathsf{abortinfo}_q$ records the public abort point and visible transcript shape before abort.
\end{definition}

Let $\mathsf{pp}$ denote the public parameters generated during setup, including public verification keys, public cryptographic descriptions, session identifiers, and public structural information.

\begin{theorem}[Privacy]
\label{thm:privacy-collusion-server}
Under the stated adversarial, setup, and cryptographic assumptions, for every PPT adversary $\mathcal A$ corrupting one server $\server{q}$ and any set of corrupted sensors within the fault threshold tolerated by the underlying fusion algorithm, there exists a PPT simulator $\mathcal S$ such that
\[
(\mathsf{pp},\mathsf{View}_{\mathcal A,q}) \approx_c \mathcal S(1^\kappa,L^{\mathsf{coal}}_{\mathcal A,q}).
\]
Hence, the corrupted coalition learns neither honest sensors' measurements nor the plaintext fusion output beyond the coalition leakage $L^{\mathsf{coal}}_{\mathcal A,q}$.
\end{theorem}
\begin{proof}[sketch]
The proof is by simulation. Given $L^{\mathsf{coal}}_{\mathcal A,q}$, the simulator constructs a simulated outcome vector $\widetilde O$ with $\mathsf{type}(\widetilde O)=\mathsf{type}(O)$. For honest-sensor $\mathsf{Hon}$ entries, it samples simulated submitted vectors $\widetilde r_i$ from the pseudorandom label distribution and uses them consistently in the simulated submissions, reports, evidence sets, first-layer checking gates, and filter-gate inputs; corrupted sensors and the corrupted server generate their messages by running $\mathcal A$. The PBFT-based transcript is generated from this simulated submission--report--evidence chain, with prepare/commit messages bound to $H(\widetilde o_i,\widetilde E_i)$. The simulator then invokes the garbled-circuit obliviousness simulator only for the final garbled fusion circuit, obtaining $(\widetilde{\widehat C}_{\mathsf{fus}},\widetilde X_{\mathsf{fus}},\widetilde Y_{\mathsf{fus}})$, and coherently regenerates the filter-gate outputs, reconstruction shares, second-layer checking gates, and output-label messages so that reconstruction yields exactly $\widetilde X_{\mathsf{fus}}$. For complementary circuit-input labels, the simulator replaces the coalition's visible shares with independently generated uniform field-element strings of the same form. Indistinguishability follows from the privacy guarantee of Shamir secret sharing, PRG security for honest sensor-side labels, garbled-circuit obliviousness for the final fusion circuit, signature unforgeability, collision resistance, and garbled-gate authenticity. Hence the simulated view is computationally indistinguishable from $(\mathsf{pp},\mathsf{View}_{\mathcal A,q})$.
\qed
\end{proof}

Appendix~\ref{app:securityproofs} gives the full simulation proof for Theorem~\ref{thm:privacy-collusion-server}.

\begin{theorem}[Correctness with Explicit Abort]
\label{thm:correctness-explicit-abort}
Under the stated adversarial and cryptographic assumptions, the honest client either outputs $\bot$ or accepts the correct output-label vector for the uniquely determined circuit-input label vector, except with negligible probability.
\end{theorem}

\begin{proof}
PBFT safety and the deterministic evidence rule imply that all honest servers that continue the protocol use the same outcome vector $O$. The first-layer checking gates are deterministic, so the honest servers derive the same final status vector. Since the client releases filter-gate labels only after receiving a three-server majority on a status vector, and at most one server is Byzantine, the accepted vector $u^\star$ equals the honest servers' final status vector. Therefore, all honest servers evaluate the filter gates according to the same status vector and generate shares for the same uniquely determined circuit-input label vector. The reconstruction procedure either recovers this vector and validates it through the second-layer checking gates, or aborts. Once reconstruction succeeds, all honest servers evaluate the same final garbled circuit on the same input-label vector. A Byzantine server cannot create a different three-server majority, and garbled-circuit authenticity prevents the client from accepting an invalid output-label vector except with negligible probability.
\qed
\end{proof}

\begin{theorem}[Liveness]
\label{thm:live}
Under the stated adversarial assumptions and PBFT liveness conditions, the honest client eventually outputs either a decoded fusion result or $\bot$.
\end{theorem}
\begin{proof}

PBFT liveness ensures that each sensor-submission agreement instance eventually decides an evidence-supported outcome, unless an explicit abort condition is triggered. If an abort is triggered, the client outputs $\bot$. Otherwise, after all agreement instances terminate, the remaining phases consist of finite local computations and finite message exchanges: first-layer checking, status reporting, client-side label release, filter-gate evaluation, share exchange, reconstruction, second-layer checking, final garbled-circuit evaluation, and output submission. Each phase either completes or triggers an explicit abort. If no abort is triggered, the client eventually receives three matching output-label vectors from the honest servers and decodes the accepted vector. Hence, the client eventually outputs either the decoded fusion result or $\bot$.
\qed
\end{proof}

\section{Evaluation}
\label{sec:evaluation}

We implemented CRSF in \texttt{C++}. Following PG~\cite{jin2024pg}, we evaluate five
Marzullo-style fault-tolerant fusion circuits in which each sensor contributes
a 16-bit value, and scale the number of sensors from 3 to at most 261. We used the same circuit files from the open source repository of PG\footnote{https://doi.org/10.6084/m9.figshare.25669026.v2}; for more detailed descriptions of these algorithms, please refer to~\cite{jin2024pg}. For authenticated
sensor submissions and inter-server PBFT messages, our implementation uses
SHA-256 for message digests and Ed25519 detached signatures via libsodium.
Concretely, we first compute a SHA-256 digest of each message and then generate
or verify an Ed25519 detached signature on that digest.

All distributed measurements were obtained on Google Cloud Platform in region
\texttt{us-central1}, zone \texttt{us-central1-a} (Iowa). Each virtual machine
is an \texttt{e2-standard-2} instance with 2~vCPUs, 8~GiB RAM, and a 40~GB
balanced persistent disk, running Ubuntu~22.04~LTS. CRSF experiments use five
VMs: four server nodes and one client-side node hosting the garbler
and software-emulated sensor processes. The PG baseline uses two VMs, one for
the client-side processes and one for the single server. All parties
communicate over the VPC network within the same zone.

Our experiments use software-emulated sensors rather than physical sensing
devices. Each emulated sensor produces the same protocol-level object as a real
sensor in CRSF: a 16-bit input value and an authenticated submission. The
evaluation therefore measures the online cost of the CRSF protocol execution,
including authenticated submission processing, PBFT-based agreement,
label reconstruction, and garbled-circuit evaluation. It does not model physical
sensor acquisition time, wireless-link variability, or embedded-device energy
consumption. All reported runtimes exclude the offline preparation phase.

\begin{figure}[t]
    \centering
    \includegraphics[width=\linewidth]{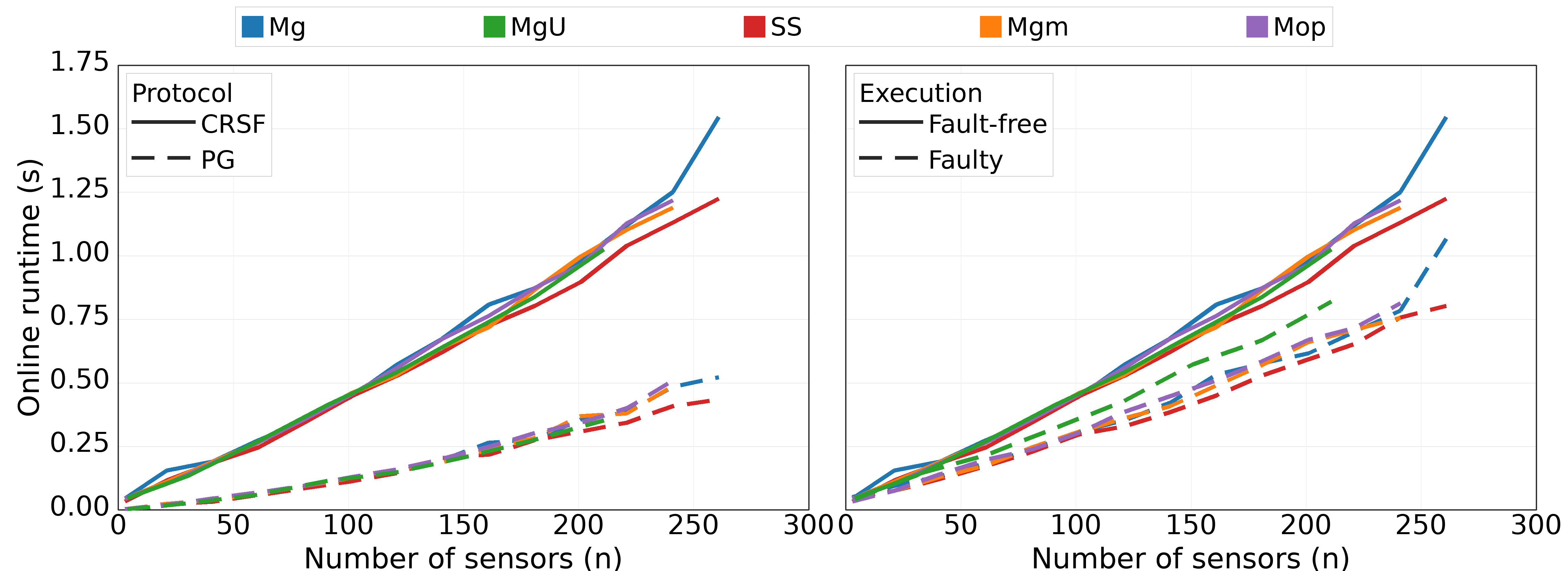}
    \caption{Online runtime of PG and CRSF under different execution settings.
    Colors denote the fusion algorithms. In the left plot, solid and dashed lines
    denote CRSF and PG, respectively. In the right plot, solid and dashed lines
    denote fault-free and faulty CRSF executions, respectively.}
    \label{fig:online-runtime}
\end{figure}

The left plot of Fig.~\ref{fig:online-runtime} compares the online runtime of PG and CRSF
as the number of sensors increases. Across all five fusion circuits, PG remains faster because it uses a single evaluator and does not incur the
inter-server PBFT agreement and threshold share reconstruction steps introduced
by CRSF. For a fixed four-server deployment, the runtime of CRSF grows approximately linearly with the number of
sensors, indicating predictable scaling despite these additional robustness
mechanisms.

The right plot of Fig.~\ref{fig:online-runtime} compares CRSF under fault-free and faulty executions. In the evaluated faulty setting, we instantiate the maximum tolerated sensor faults by having the corresponding faulty sensors omit their submissions, and instantiate one server fault by making one backup server non-responsive. After the non-responsive backup is detected, the primary proceeds with the remaining responsive servers, which are sufficient to form the required quorum. In this setting, the responsive servers process fewer sensor-submission records and perform fewer signature verifications during the agreement and validation steps. This explains why the measured online runtime can be lower than in the corresponding fault-free execution. This faulty execution represents one evaluated adversarial scenario rather than
an exhaustive fault-injection campaign. Overall, CRSF trades higher online runtime for Byzantine-resilient distributed
evaluation and robustness against sensor--server collusion.

\begin{table*}[t]
  \centering
  \caption{Communication volume and measured phase wall time at the largest tested circuit per algorithm.}
  \label{tab:crsf-comm-bytes-latency}
  \setlength{\tabcolsep}{5.2pt}
  \renewcommand{\arraystretch}{1.05}
  \begin{tabular}{@{}llc
      S[table-format=5.0,group-separator={\,}]
      S[table-format=3.1]
      S[table-format=4.0,group-separator={\,}]
      S[table-format=3.1]
      S[table-format=5.0,group-separator={\,}]
      S[table-format=4.1]@{}}
    \toprule
    & & & \multicolumn{2}{c}{BFT agreement}
          & \multicolumn{2}{c}{Reconstruction}
          & \multicolumn{2}{c}{Total} \\
    \cmidrule(lr){4-5}\cmidrule(lr){6-7}\cmidrule(l){8-9}
    Circuit & Scenario & $n$
      & {KB} & {ms}
      & {KB} & {ms}
      & {KB} & {ms} \\
    \midrule
    \multirow{2}{*}{Mg}  & Fault-free & \multirow{2}{*}{261} & 9891 & 488.0 & 1566 & 115.7 & 15492 & 1546.0 \\
                         & Faulty     &                      & 8365 & 474.9 &  783 &  39.0 & 12424 & 1066.3 \\
    \addlinespace[1pt]

    \multirow{2}{*}{MgU} & Fault-free & \multirow{2}{*}{211} & 7998 & 426.4 & 1266 & 122.9 & 12526 & 1025.2 \\
                         & Faulty     &                      & 6861 & 369.8 &  633 &  39.9 & 10181 &  819.4 \\
    \addlinespace[1pt]

    \multirow{2}{*}{SS}  & Fault-free & \multirow{2}{*}{261} & 9891 & 506.5 & 1566 & 138.2 & 15492 & 1224.2 \\
                         & Faulty     &                      & 8365 & 353.0 &  783 &  47.5 & 12424 &  802.2 \\
    \addlinespace[1pt]

    \multirow{2}{*}{Mgm} & Fault-free & \multirow{2}{*}{241} & 9134 & 481.7 & 1446 & 128.6 & 14305 & 1189.7 \\
                         & Faulty     &                      & 7725 & 323.5 &  723 &  34.3 & 11473 &  754.9 \\
    \addlinespace[1pt]

    \multirow{2}{*}{Mop} & Fault-free & \multirow{2}{*}{241} & 9134 & 472.3 & 1446 & 139.3 & 14305 & 1217.8 \\
                         & Faulty     &                      & 7725 & 325.9 &  723 &  45.7 & 11473 &  813.1 \\
    \bottomrule
  \end{tabular}
\end{table*}

Table~\ref{tab:crsf-comm-bytes-latency} further breaks down the main
CRSF-specific online costs at the largest tested instance for each fusion
algorithm. The PBFT-based submission agreement accounts for the largest
communication volume, ranging from about 7.8--9.7~MB in fault-free executions,
because each server exchanges authenticated evidence for batched sensor
submissions. Threshold share reconstruction is smaller in communication volume,
about 1.2--1.5~MB in the fault-free cases, but still contributes measurable
wall time because the servers exchange and validate shares before evaluating the
garbled fusion circuit. Across the five algorithms, the measured wall time of
the BFT agreement phase is below 0.51 seconds, and the reconstruction phase is
below 0.14 seconds in the fault-free setting.
Faulty executions have lower communication volume and wall time in this table, which is consistent with the performance trend shown in Fig.~\ref{fig:online-runtime}.

\section{Related Work}\label{sec:RelatedWork}

\noindent\textbf{Collusion-resilient sensing and aggregation.} Prior work has studied collusion resistance in several sensing and aggregation settings. Shi et al.~\cite{shi2010prisense} propose PriSense, a privacy-preserving aggregation scheme for people-centric urban sensing based on data slicing and mixing. PriSense supports additive and non-additive statistical queries and protects user privacy against a threshold number of colluding users and aggregation servers. Günther et al.~\cite{gunther2014privacy} revisit the PEPSI participatory-sensing model and show that cross-party collusion, such as collusion between a mobile node and a querier or between the service provider and a mobile node, can break earlier privacy guarantees. Their PEPSICo construction adds collusion resistance and supports optional aggregation for small sensor readings using an additively homomorphic identity-based encryption scheme. Yao et al.~\cite{yao2014privacy} study privacy-preserving SUM aggregation in two-tiered mobile wireless sensor networks and propose PDAAS and PDACAS. In particular, PDACAS protects the raw readings of mobile sensor nodes even when the cell header, acting as the aggregator, and the sink collude, while still allowing the sink to derive the SUM. These works motivate the relevance of collusion in sensing systems, but they differ from CRSF in both functionality and leakage surface. PriSense and PDACAS are privacy-preserving aggregation protocols: they hide individual users' or sensors' readings while enabling an aggregation server or sink to compute the intended aggregate. PEPSICo addresses a different participatory-sensing problem.
Complementary work studies secure and verifiable aggregation against malicious
aggregators and corrupted users or sensors~\cite{przydatek2003sia,palazzo2024privacy}.
These works address aggregation correctness or verifiable summation, whereas
CRSF targets garbled-circuit-based sensor fusion with sensor--server collusion
and Byzantine-robust participation agreement.

\vspace{2mm}
\noindent\textbf{Secure aggregation and federated learning.}
Our threat model and technical setting are fundamentally different from those considered in secure aggregation for federated learning~\cite{fereidooni2021safelearn,pasquini2022eluding,zhang2023safelearning,zhao2021sear}. Classical secure aggregation protocols are designed to hide each client’s local model update while allowing the server to recover an aggregate. These formulations usually involve only two roles, namely clients and a server, and the aggregate is intentionally revealed to the server. For example, efficient dropout-resilient aggregation~\cite{liu2022efficient} proposes a scalable protocol for secure sum computation, where secret sharing is used primarily to recover masks under client dropouts, and the server learns the final aggregate. FLock~\cite{chen2025flock} instead adopts a more MPC-style design with multiple aggregators, enabling richer aggregation functionalities at the cost of higher system complexity. In contrast, CRSF targets three-role sensor fusion, where the servers should learn neither individual sensor measurements nor the final fusion result, even under sensor--server collusion and Byzantine server behavior.

\section{Conclusion}\label{sec:conclusion}
We presented CRSF, a collusion-resilient privacy-preserving sensor-fusion
protocol for the sensor--server--client architecture.
It combines PBFT-based submission agreement with status-dependent label release
and byte-wise $(3,4)$ Shamir sharing of circuit-input labels, preventing a
Byzantine server from unilaterally controlling participation and preventing an
admissible sensor--server coalition from reconstructing complementary
circuit-input labels. We proved privacy, correctness with explicit abort, and liveness under the
stated assumptions. Our implementation and Google Cloud evaluation show that
CRSF introduces additional online runtime from agreement and reconstruction, but
retains predictable scaling across the evaluated fault-tolerant fusion circuits.
These results show that stronger robustness against sensor--server collusion and
Byzantine participation manipulation can be achieved with moderate online
overhead.

\section*{Acknowledgment}
Chenglu Jin is (partially) supported by project CiCS of the research programme Gravitation, which is (partly) financed by the Dutch Research Council (NWO) under the grant 024.006.037. Chao Yin is supported by the China Scholarship Council (CSC) and the Dutch Sector Plan. This work is supported in part by the National Natural Science Foundation of China under 62272043 and Yangtze Delta Region Institute of Tsinghua University, Zhejiang (No.LZZLX24F007).

\bibliographystyle{splncs04}
\bibliography{mybibliography}

\appendix
\section{Security Proofs}
\label{app:securityproofs}

\subsection{Proof of Theorem~\ref{thm:privacy-collusion-server}}
\label{app:proof-privacy-collusion-server}

\begin{proof}
We construct a simulator $\mathcal S$ and prove indistinguishability through two hybrids.

\paragraph{Simulator.}
Given $L^{\mathsf{coal}}_{\mathcal A,q}$, the simulator samples simulated setup material and runs $\mathcal A$ with $\mathsf{inp}_{\mathcal A,q}$ and $\mathsf{rand}_{\mathcal A,q}$. It constructs a simulated outcome vector $\widetilde O$ such that
$\mathsf{type}(\widetilde O)=\mathsf{type}(O)$.
For every honest-sensor entry whose leaked type is $\mathsf{Hon}$, it samples a fresh simulated submitted vector $\widetilde r_i$ from the pseudorandom label distribution and sets
$\widetilde o_i=(\mathsf{Hon},\widetilde r_i)$. The same $\widetilde r_i$ is used consistently in the simulated submissions, local reports, evidence sets, first-layer checking gates, and filter-gate inputs. For corrupted sensors and the corrupted server, all submissions, reports, and proposals are induced by running $\mathcal A$.

The PBFT-based agreement transcript is generated from the simulated submission layer: the simulator forms signed honest-server local reports, constructs or validates evidence sets $\widetilde E_i$, and derives
$\widetilde o_i=\mathsf{Derive}(\widetilde E_i)$
by the same deterministic evidence rule as the real protocol, and generates honest prepare/commit messages on
$H(\widetilde o_i,\widetilde E_i)$.
It then generates the first-layer checking behavior, signed status-vector messages, and client-side label-release transcripts consistently with $u^\star$. The client releases to the corrupted server only the server-specific client-side filter-gate labels corresponding to the accepted majority status vector $u^\star$.

The simulator invokes the garbled-circuit obliviousness simulator for the final garbled fusion circuit. It shares every simulated circuit-input label used in the simulated fusion-circuit evaluation with byte-wise $(3,4)$ Shamir sharing, and constructs the server-specific filter gates, reconstruction transcript, second-layer checking gates, and final output-label messages so that reconstruction yields exactly these labels. For each complementary circuit-input label, any share value visible to the corrupted server is generated as an independent uniform field-element string of the same form, as justified by below-threshold byte-wise Shamir privacy.

If $\mathsf{abortinfo}_q$ specifies an abort, the simulator outputs only the visible transcript prefix up to that abort point.

\paragraph{Hybrids.}
\begin{itemize}[leftmargin=*,itemsep=2pt,topsep=2pt]
    \item Let $H_0=(\mathsf{pp},\mathsf{View}_{\mathcal A,q})$ be the real distribution.
    \item In $H_1$, replace the corrupted coalition's visible below-threshold shares of complementary circuit-input labels by uniformly random byte-wise shares. This does not change the labels used in evaluation, their reconstruction shares, checking outcomes, or the final fusion-circuit evaluator view. By byte-wise $(3,4)$ Shamir privacy, $H_0\approx_s H_1$.

    \item In $H_2$, replace the real honest-side interface by the simulated one described above. Honest sensor-side labels visible in submissions, evidence, checking gates, and filter-gate inputs are replaced by simulated pseudorandom labels, which is indistinguishable by PRG security. The PBFT transcript remains consistent because signatures are generated under simulated honest signing keys, corrupted-party messages are induced by $\mathcal A$, and prepare/commit digests bind the simulated outcomes and evidence by collision resistance of $H$. For the final fusion circuit, $H_2$ replaces the real evaluator view $(\widehat C_{\mathsf{fus}},X_{\mathsf{fus}},Y_{\mathsf{fus}})$ with a simulated one generated as follows: $(\widetilde{\widehat C}_{\mathsf{fus}},\widetilde X_{\mathsf{fus}})
\leftarrow
\mathsf{Sim}^{\mathsf{obv}}_{\mathsf{GC}}(1^\kappa,\Phi_{\mathsf{fus}}),
\widetilde Y_{\mathsf{fus}} \leftarrow
\mathsf{Ev}(\widetilde{\widehat C}_{\mathsf{fus}},\widetilde X_{\mathsf{fus}})$. This replacement is computationally indistinguishable by the obliviousness of the garbling scheme for the final fusion circuit.
The surrounding filter-gate outputs, reconstruction shares, and second-layer checking gates are regenerated coherently so that reconstruction yields exactly $\widetilde X_{\mathsf{fus}}$.
Auxiliary checking and filter gates are generated so that every row evaluable by the corrupted coalition has the prescribed output. Garbled-gate authenticity prevents useful outputs from invalid labels except with negligible probability. Therefore, $H_1\approx_c H_2$.
\end{itemize}

By construction, $H_2$ is exactly the output distribution of
$\mathcal S(1^\kappa,L^{\mathsf{coal}}_{\mathcal A,q})$. Hence $(\mathsf{pp},\mathsf{View}_{\mathcal A,q})
=H_0
\approx_s H_1
\approx_c H_2
=
\mathcal S(1^\kappa,L^{\mathsf{coal}}_{\mathcal A,q})$, which proves the theorem.
\end{proof}

\end{document}